\documentclass[twocolumn,superscriptaddress,nofootinbib,preprintnumbers,showpacs,tightenlines,notitlepage]{revtex4}
\usepackage[latin1]{inputenc}
\usepackage{graphicx}
\usepackage{amssymb}
\usepackage{color}
\usepackage{float}
\usepackage{amsmath}
\usepackage{amsfonts}
\usepackage{dcolumn}
\usepackage{hyperref}
\usepackage{amsthm}
\usepackage{color}

\def\3nab{\tilde{\nabla}}

\def\la {\langle}
\def\ra {\rangle}

\def\be {\begin{equation}}
\def\ee {\end{equation}}
\def\ba {\begin{eqnarray}}
\def\ea {\end{eqnarray}}
\newtheorem{lem}{Lemma}
\newtheorem{prop}{Proposition}

\newcommand{\barray}{\begin{array}}
\newcommand{\earray}{\end{array}}

 \newcommand{\nab}{\nabla}

\newcommand \ep {\epsilon}

\newcommand \om {\omega}

\begin{document}

\title{Matter shear and vorticity in conformally flat spacetimes}

\author{Roger Mayala}
\email{rmayala2@gmail.com }
\affiliation{Astrophysics Research Centre, School of Mathematics, Statistics and Computer Science, University of KwaZulu-Natal, Private Bag X54001, Durban 4000, South Africa.}
\author{Rituparno Goswami}
\email{Goswami@ukzn.ac.za}
\affiliation{Astrophysics Research Centre, School of Mathematics, Statistics and Computer Science, University of KwaZulu-Natal, Private Bag X54001, Durban 4000, South Africa.}
\author{Sunil D. Maharaj}
\email{Maharaj@ukzn.ac.za}
\affiliation{Astrophysics Research Centre, School of Mathematics, Statistics and Computer Science, University of KwaZulu-Natal, Private Bag X54001, Durban 4000, South Africa.}

\begin{abstract}
In this paper we consider conformally flat perturbations on the Friedmann Lemaitre Robertson Walker (FLRW) spacetime containing a general matter field. Working with the linearised field equations, we unearth some important geometrical properties of matter shear and vorticity and how they interact with the thermodynamical quantities in the absence of any free gravity powered by the Weyl curvature. As there are hardly any physically realistic rotating exact conformally flat solutions in general relativity, these covariant and gauge invariant results bring out transparently the role of vorticity in the linearised regime.  Most interestingly, we demonstrate that the matter shear obeys a transverse traceless tensor wave equation, and the vorticity obeys a vector wave equation in this regime. These shear and vorticity waves replace the gravitational waves in the sense that they causally carry the information about local change in the curvature of these spacetimes.

\end{abstract}

\pacs{04.20.Cv	, 04.20.Dw}

\maketitle
\section{Introduction}
Conformal flatness is a condition that is often applied in the study of gravitational interactions, since many of these models characterise spacetimes of physical importance, (e.g. Friedmann Lemaitre Robertson Walker (FLRW)). Several other classes of conformally flat spacetimes have been applied in cosmology, including generalised Friedmann models, generalised Schwarzschild interior models, Bertottic-Robinson models and radiation fields \cite{Stephani}. In a conformally flat spacetime the Weyl tensor vanishes identically \cite{HE} and the technique of embedding has proved to be a useful tool in generating a variety of exact solutions \cite{Krasinski,Ellisbook}, with perfect fluids, pure radiation and electromagnetic fields. Conformal flatness is also widely used in studying gravitational collapse for various matter fields. Collapse in the presence of scalar fields was studied in \cite{Narayan}, with dissipative matter giving rise to radiating stellar configurations \cite{Herrera1,Herrera2,Suri,Maharaj,Ranjan}. Conformal flatness has also proven to be useful in constructing static anisotropic stars that can represent real pulsars \cite{Ivanov}. In addition to this, vanishing of the Weyl tensor helps to solve the field equations in modified gravity theories (see for example \cite{Soumya} for $f(R)$-theories of gravity, and \cite{Sudan} for Einstein-Gauss-Bonnet (EGB) gravity).\\
 
Although a lot of works have been done on the conformal symmetries of these spacetimes and also numerous solutions have been found, most of these stick to spacetimes with specific symmetries (for example, spherical or cylindrical symmetries) and very specific types of matter fields. The main reason behind this is, in spite of Weyl tensor being identically zero in these spacetimes, the field equations still remain extremely complicated for any general treatment. To overcome this hurdle, we start by taking baby steps. We consider a general, but conformally flat perturbations on Friedmann Lemaitre Robertson Walker (FLRW) spacetime and work with linearised field equations up to the first order. In other words, the Weyl tensor vanishes identically in the perturbed spacetime. Nevertheless, all the other quantities, that were zero in the background become first order quantities in the perturbed scenario and we deal with a general form of matter with anisotropic stresses and heat flux to the first order of smallness. The main aim of this investigation is to track the behaviour of geometrical quantities, like matter shear or vorticity in the perturbed conformally flat scenarios. This will definitely give an indication as how these quantities will behave in the most general treatment of field equations. \\

In this paper we use a local semitetrad covariant formalism, and hence all the perturbation results are frame invariant and gauge invariant in general. Pioneering works in this regard was done by Bruni et. al. \cite{Marco}, where the scalar, vector and tensor modes of the perturbations were treated in a covariant and gauge invariant way. Further works on cosmological gravitational waves were done in \cite{Dunsby,Roy}. Similar techniques were used in \cite{Shear1,Shear2}, to show that in a shear-free perturbation of FLRW spacetime, matter cannot expand and rotate simultaneously even in the linearised regime. \\

Unless otherwise specified, we use natural units ($c=8\pi G=1$) throughout this paper, Latin indices run from 0 to 3. 
The symbol $\nabla$ represents the usual covariant derivative and $\partial$ corresponds to partial differentiation. 
We use the $(-,+,+,+)$ signature and the Riemann tensor is defined by
\begin{equation}
R^{a}{}_{bcd}=\Gamma^a{}_{bd,c}-\Gamma^a{}_{bc,d}+ \Gamma^e{}_{bd}\Gamma^a{}_{ce}-\Gamma^e{}_{bc}\Gamma^a{}_{de}\;.
\end{equation}
The Ricci tensor is obtained by contracting the {\em first} and the {\em third} indices
\begin{equation}\label{Ricci}
R_{ab}=g^{cd}R_{cadb}\;.
\end{equation}
The symmetrisation and the anti-symmetrisation  over the indices of a tensor are defined by
\begin{equation}
T_{(a b)}= \frac{1}{2}\left(T_{a b}+T_{b a}\right)\;,\qquad T_{[a b]}= \frac{1}{2}\left(T_{a b}-T_{b a}\right)\,.
\end{equation}
The Hilbert--Einstein action in the presence of matter is given by
\begin{equation}
{\cal S}=\frac12\int d^4x \sqrt{-g}\left[R-2{\cal L}_m \right]\;,
\end{equation}
variation of which gives the Einstein field equations as
\be
G_{ab} =T_{ab}\;.
\ee

\section{1+3 covariant splitting of spacetimes}
This formalism \cite{EllisCovariant} is based on a local 1+3 threading of the spacetime manifold 
and has been a very handy tool for understanding many geometrical and physical aspects of relativistic fluid flows, both in general relativity or in the gauge invariant, 
covariant perturbation formalism \cite{Ellisbook}.
We first define a timelike congruence with a unit tangent vector $u^a$ along the fluid flow lines. Then the spacetime is split
locally in the form $R\otimes V$ where  $R$ denotes the worldline
along $u^a$ and  $V$ is the 3-space perpendicular to $u^a$. Any
vector $X^a$  can then be projected on the 3-space by the projection
tensor $h^a{}_b=g^a{}_b+u^au_b$.  The choice of $u^a$ naturally defines 
two derivatives: the \textit{covariant
time derivative} along the observers' worldlines (denoted by a dot), and the fully orthogonally
\textit{projected covariant derivative} $D$ on the three-dimensional space. 
For any tensor $ S^{a..b}{}_{c..d}$, we have
\be
\dot{S}^{a..b}{}_{c..d}{} = u^{e} \nab_{e} {S}^{a..b}{}_{c..d} \;,
\ee
and 
\be D_{e}S^{a..b}{}_{c..d}{} = h^a{}_f
h^p{}_c...h^b{}_g h^q{}_d h^r{}_e \nab_{r} {S}^{f..g}{}_{p..q}\;,
\ee 
with total projection on all the free indices.  Angle brackets
denote orthogonal projections of vectors, and the orthogonally
\textit{projected symmetric trace-free} PSTF part of tensors: 
\be
V^{\la a \ra} = h^{a}{}_{b}V^{b}~, ~ S^{\la ab \ra} = \left[
h^{(a}{}_c {} h^{b)}{}_d - \frac{1}{3} h^{ab}h_{cd}\right] S^{cd}\;.
\label{PSTF} 
\ee 
This also defines the 3-volume element 
\be \ep_{a b c}=-\sqrt{|g|}\delta^0_{\left[ a
\right. }\delta^1_b\delta^2_c\delta^3_{\left. d \right] }u^d\;.
\label{eps1} 
\ee 
The covariant derivative of the timelike vector $u^a$ can now be
decomposed into the irreducible part as 
\be
\nabla_au_b=-\dot{u}_au_b+\frac13h_{ab}\Theta+\sigma_{ab}+\ep_{a b
c}\om^c\;, 
\ee 
where $\dot{u}_a$ is the acceleration,
$\Theta=D_au^a$ is the expansion, $\sigma_{ab}=D_{\la a}u_{b \ra}$
is the shear tensor and $\om^a=\ep^{a b c}D_bu_c$ is the vorticity
vector. Similarly the Weyl curvature tensor can be decomposed
irreducibly into the Gravito-Electric and Gravito-Magnetic parts as
\be 
E_{ab}=C_{abcd}u^cu^d=E_{\la ab\ra}\;;\;
H_{ab}=\frac12\ep_{acd}C^{cd}{}_{be}u^e=H_{\la ab\ra}\;, 
\ee 
which allows for a covariant description of tidal forces and gravitational
radiation. The energy momentum tensor for a general matter field can 
be similarly decomposed as follows
\be
T_{ab}=\mu u_au_b+q_au_b+q_bu_a+ph_{ab}+\pi_{ab}\;,
\ee
where $\mu=T_{ab}u^au^b$ is the energy density, $p=(1/3 )h^{ab}T_{ab}$ is the isotropic pressure, $q_a=q_{\la a\ra}=-h^{c}{}_aT_{cd}u^d$ is the 3-vector defining 
the heat flux and $\pi_{ab}=\pi_{\la ab\ra}$ is the anisotropic stress. 
\section{Conformally flat perturbation around FLRW spacetime}

To see in a transparent manner, how the absence of the Weyl tensor affects other geometrical and thermodynamic quantities, we consider a conformally flat perturbation of FLRW manifold. In other words, the background metric is given as 
\be\label{FLRW}
ds^2=-dt^2+\frac{a^2(t)}{1-kr^2}dr^2+r^2a^2(t)(d\theta^2+\sin^2\theta d\phi^2)\;.
\ee
One can easily see that the non-zero geometric and thermodynamic quantities for the background are
\be
\mathcal{D}_0=\{\Theta, \mu, p\}\;.
\ee
We now perturb this background spacetime in such a way that the perturbed spacetime still remains conformally flat, that is the Weyl tensor remains identically zero. In that case the quantities that are of first order smallness in the perturbed spacetime are given as
\be
\mathcal{D}_1=\{\dot{u}_{\la a\ra}, \om_{\la a\ra},\sigma_{\la ab\ra}, q_{\la a\ra}, \pi_{\la ab\ra} \}\;.
\ee
The Riemann tensor of the perturbed spacetime can now be completely specified in therms of the matter variables as follows:
\ba
R^{ab}{}{}_{cd} & = &-2\left(u^{[a}h^{b]}{}_{[c}q_{d]}+u_{[c}h^{[a}{}_{d]}q^{b]}\right.\nonumber\\
&&+\left.u^{[a}u_{[c}\pi^{b]}{}_{d]}- h^{[a}{}_{[c}\pi^{b]}{}_{d]}\right)\nonumber\\
&&+\frac{2}{3}\left[(\mu +3p)u^{[a}u_{[c}h^{b]}{}_{d]}+\mu h^a{}_{[c}h^b{}_{d]}\right]\;.
\ea
We now use this form of the Riemann tensor to get the Ricci identities of the vector $u^a$ and doubly contracted Bianchi identities (linearised by setting any higher power of the first order quantities to zero). These 
equations can be further classified into evolution (time derivative) equations and constraints on 3-space. Solutions to these equation will then completely specify the dynamics of the perturbed spacetimes to the linear order and furthermore all these equations remain gauge invariant by Stewart and Walker lemma \cite{SW}.
\subsection{Evolution equations}
The evolution equations take the following form
\be\label{E1}
\dot{\Theta}-D_a\dot{u}^a = -\frac{1}{3}\Theta^2 -\frac{1}{2}(\mu + 3p),
\ee
\be\label{E2}
\dot{\sigma}^{\la ab \ra}-D^{\la a}\dot{u}^{b\ra}=\frac{1}{2}\pi^{ab} -\frac{2}{3}\Theta\sigma^{ab},
\ee
\be\label{E3}
\dot{\omega}^{\la a\ra}-\frac{1}{2}\epsilon^{abc}D_b\dot{u}_c = -\frac{2}{3}\Theta\omega^a,
\ee
\be\label{E4}
\dot{\pi}^{<ab>}+D^{<a}q^{b>}
= -(\mu +p)\sigma^{ab}- \frac{\Theta}{3}\pi^{ab},
\ee
\be\label{E5}
\dot{q}^{\la a\ra}+D^ap +D_b \pi^{ab}=-\frac{4}{3}\Theta q^a-(\mu +p)\dot{u}^a ,
\ee
\be\label{E6}
\dot{\mu} +D_aq^a= -\Theta(\mu +p).
\ee
\subsection{Constraints}
The constraints equation on a given spatial 3-surface can be written as
\be\label{C1}
(C_1)^a \equiv D_b\sigma^{ab}-\frac{2}{3}D^a\Theta+\epsilon^{abc}D_b\omega_c+q^a=0,
\ee
\be\label{C2}
(C_2) \equiv D_a\omega^a=0 ,
\ee
\be\label{C3}
(C_3)^{ab}\equiv D^{<a}\omega^{b>}-\epsilon^{cd<a}D_c\sigma^{b>}{}_d=0,
\ee
\be\label{C4}
(C_4)^a\equiv \frac{1}{2}D_b\pi^{ab}-\frac{1}{3}D^a\mu +\frac{1}{3}\Theta q^a=0,
\ee
\be\label{C5}
(C_5)^a\equiv (\mu +p)\omega^a +\frac{1}{2}\epsilon^{abc}D_bq_c =0,
\ee
\be\label{C6}
(C_6)^{ab}\equiv \epsilon^{cd<a}D_c \pi^{b>}{}_d = 0.
\ee
We note that the last constraint (\ref{C6}) is not an original constraint of the field equations. We get this constraint by forcing the perturbed spacetime to be conformally flat. A consistent evolution of this constraint (so that this is valid at all epochs) will give further restrictions on different geometrical and thermodynamical quantities as we shall see in the next section. 

\subsection{Commutation relations}
For a linearised model about an FLRW spacetime, we have the following commutation relations between the derivative operators. For any scalar function $f$,
\ba 
D_{[a}D_{b]}{f} &=& \epsilon_{abc}\omega^c\dot{f},\label{com11}\\
\epsilon^{abc}D_bD_cf &= &2\omega^a\dot{f},\label{com12}\\
h^a{}_b[D^bf]\dot{}&= &D^a\dot{f}-\frac{1}{3}\Theta D^af. \label{com13}
\ea
Also for any first order $3$-vector $V^a$, we have
\ba 
h^a{}_ch^b{}_d[D^cV^d]\dot{}&=& D^a\dot{V}^{<b>}-\frac{1}{3}\Theta D^aV^b,\label{com21}\\ 
(D_aV^a)\dot{}&=& D_a\dot{V}^{<a>}-\frac{1}{3}\Theta D_aV^a,\label{com22}\\
D_{[a}D_{b]}V_c &=& \frac{1}{3}\left(\mu - \frac{1}{3}\Theta^2\right)h_{c[a}V_{b]},\label{com23}\\
h^a{}_b(\epsilon^{bcd}D_cV_d)\dot{} &= &\epsilon^{abc}D_b\dot{V}_{<c>}-\frac{1}{3}\Theta \epsilon^{abc}D_bV_c.\label{com24}
\ea
Similarly, for any first order second rank $3$-tensor $A^{ab}$, we have 
\ba 
D_{[a}D_{b]}A^{cd} &= &\frac{2}{3}\left(\mu - \frac{1}{3}\Theta^2\right)h_{\;\;[a}^{(c}A^{d)}{}_{b]},\label{com31}\\ 
h^a{}_ch^b{}_d[\epsilon^{ef<c}D_eA^{d>}{}_f]\dot{}&=&\epsilon^{dc<a}D_c \dot{A}^{<b>e>}h_{de}\nonumber\\
&&-\frac{1}{3}\Theta\epsilon^{cd<a}D_cA^{b>}{}_d.\label{com32}
\ea
For brevity, we will use the following notation in this paper (for any 3-vector $V$ and second rank 3-tensor $A_{ab}$)
\ba
{\rm div}\, V=D_aV^a, & ({\rm curl}\, V)_a=\epsilon_{abc}D^bV^c,\\
({\rm div}\, A)_a=D^bA_{ab}, & ({\rm curl}\, A)_{ab}=\epsilon_{cd\la a}D^cA^d_{\;b\ra}.
\ea
Further to the above, there are some important identities of first order vectors and tensors in perturbed FLRW spacetime, which we list below \cite{Roy}
\be\label{id1}
D^a({\rm curl}\, V)_a=0\,,
\ee
\be\label{id2}
D^b({\rm curl}\, A)_{ab}=\frac12{\rm curl}\,(D^bA_{ab})\,,
\ee
\ba\label{id3}
({\rm curl}\, {\rm curl}\,V)_a&=&-D^2V_a+D_a(D^bV_b)\nonumber\\
&&+\frac23\left(\mu- \frac{1}{3}\Theta^2\right)V_a,
\ea
\ba\label{id4}
({\rm curl}\, {\rm curl}\,A)_{ab}&=&-D^2A_{ab}+\frac32D_{\la a}D^cA_{b\ra c}\nonumber\\
&&+\left(\mu- \frac{1}{3}\Theta^2\right)A_{ab}.
\ea

\section{Some geometrical results on shear and vorticity}
In this section we will state and prove several important geometrical properties of matter shear and viscosity in the perturbed, conformally flat spacetime. We would like to emphasise here, that throughout this paper we consider the matter field satisfying the weak energy conditions and that would imply $(\mu+p)$ is strictly greater than zero. We start by proving a lemma for a general perturbed FLRW spacetime, which will be used for other results. 
\begin{lem}
For a perturbed FLRW spacetime, the spatial variation tensor $D_aV_b$ for any first order 3-vector $V^a$ is curl free in the linearised regime. 
\end{lem}
\begin{proof}
From the commutation relation (\ref{com23}), we know that
\begin{equation} 
(D_cD_d-D_dD_c)V_b = \frac{1}{3}\left(\mu - \frac{1}{3}\Theta^2\right)(h_{bc}V_d-h_{bd}V_c)\;.
\end{equation}
Acting with $\epsilon^{adc}$ on both sides of the above equation, symmetrising on the indices $a,b$ and subtracting the trace, we get
\begin{equation} 
\epsilon^{dc<a}D_cD_dV^{b>} = -\frac{1}{3}\left(\mu - \frac{1}{3}\Theta^2\right)\epsilon^{<ab>d}V_d\;.
\end{equation}
The right hand side of the above equation is identically zero as $\epsilon^{abd}$ is completely antisymmetric. Hence we get the required result, $({\rm curl}\, D_aV_b)=0$.
\end{proof}
We note that although the above result is a constraint on a given hypersurface, it is consistently time propagated. To show this, we note that if $V^a$ is a first order 3-vector, then so is $\dot{V}^{\la a\ra}$. Then using the commutation relation (\ref{com24}), we see that $({\rm curl}\, D_aV_b)\dot{}$ vanishes identically. 
\begin{prop}
For a linear and conformally flat perturbation of FLRW spacetime, the shear tensor is curl free. 
\end{prop}
\begin{proof}
To prove this, we demand that the new constraint (\ref{C6}), due to the absence of Weyl, be consistently time propagated. In other words, we must have $({\rm curl}\,\pi)\dot{}_{\la ab\ra}=0$. Using the commutation relation (\ref{com32}), we then see that $({\rm curl}\,\dot{\pi})_{\la ab\ra}$ must vanish identically. Now we use the evolution equation (\ref{E4}), to get
\be
({\rm curl}\, D_{\la a} q_{b\ra})+(\mu+p)({\rm curl}\, \sigma)_{\la ab\ra}=0\;.
\ee
The first term in the LHS is zero by Lemma 1. Since the weak energy condition demands that $(\mu+p)$ is strictly greater than zero, we see that the shear tensor must be curl free. Again, this is a result on a given hypersurface. We time propagate this constraint and using (\ref{com32}) and (\ref{E2}), we see that $({\rm curl}\, \sigma)\dot{}_{\la ab\ra}=0$ is identically satisfied. Hence the shear tensor being curl free is true at all epochs.
\end{proof}
\begin{prop}
For a linear and conformally flat perturbation of FLRW spacetime, 
\begin{enumerate}
\item For non-vanishing vorticity, the heat flux vector is purely axial on a given hypersurface, and consistent time propagation of this constraint gives an implicit equation of state relating the density and isotropic pressure.
\item The vorticity is generated purely by the curl of heat flux vector for all epochs. 
 \end{enumerate}
\end{prop}
\begin{proof}
To prove the first point, we take the curl of constraint (\ref{C4}) to get
\be
\frac12 ({\rm curl}\, {\rm div}\, \pi)^a-\frac13 ({\rm curl}\, D^a\mu)+\frac13\Theta({\rm curl}\, q)^a=0\;.
\ee
The first term of above equation can be written as $2D_b({\rm curl}\, \pi)^{ab}$ (by commutation (\ref{id2})), which vanishes because of the new constraint (\ref{C6}) . Hence the above equation becomes
\be
\frac13 ({\rm curl}\, D^a\mu)-\frac13\Theta({\rm curl}\, q)^a=0\;.
\ee
Now by the relation (\ref{com12}) we see that $({\rm curl}\, D^a\mu)=2\om^a\dot{\mu}$. Further, using (\ref{C5}), we get $({\rm curl}\, q)^a=-2(\mu+p)\om^a$. Putting all of these in the above equation, and noting that the vorticity is not vanishing, we get
\be
\dot{\mu}+\Theta(\mu+p)=0\;.
\ee
Comparing this with the evolution equation (\ref{E6}), we can easily see that 
\be
D_aq^a=0\;. 
\ee
In other words the heat flux vector has vanishing divergence and hence it is purely axial. To check that this constraint is consistently time propagated, we impose the condition
\be
(D_aq^a)\dot{}=0\;.
\ee
Using the commutation relation (\ref{com22}), we see that this gives $D_a\dot{q}^{\la a \ra}=0$. Now using the evolution equation (\ref{E5}) we get 
\be\label{eq1}
D^2p+D_aD_b\pi^{ab}+(\mu+p)D_a\dot{u}^a=0\;.
\ee
Now if we take the divergence of the constraint (\ref{C4}), we get 
\be
D_aD_b\pi^{ab}=\frac23D^2\mu\;.
\ee
Substituting the above in equation(\ref{eq1}) we get the implicit equation of state as a second order differential equation,
\be
D^2\left(p+\frac23\mu\right)+(\mu+p)D_a\dot{u}^a=0\;.
\ee
In other words, for the vorticity to remain non-zero at all epochs, the above equation of state relating the energy density and isotropic pressure must be satisfied.\\

The proof of the second point is obvious from the constraint equation (\ref{C5}). To see, whether this constraint is identically satisfied at all epochs, we take a dot of this constraint and use the density evolution equation  (\ref{E6}), vorticity evolution equation (\ref{E3}) and heat flux evolution equation (\ref{E5}). One can then easily check that this constraint is consistently time propagated and hence the vorticity is purely generated by the curl of heat flux vector for all epochs. 
\end{proof}

\begin{prop}
If the spacetime is perturbed in a conformally flat way about the FLRW background, then the spatial variation tensor of the vorticity is purely antisymmetric at all epochs. The curl of the vorticity is then generated by the heat flux vector and it's Laplacian. 
\end{prop}
\begin{proof}
From the constraint equations (\ref{C2},\ref{C3}) and using the result from Proposition 1, we can immediately see that $D_a\omega^a$ and $D^{<a}\omega^{b>}$ vanish identically on a given epoch. These relations can then be time propagated to check that they remain true for all epochs. Since both the trace and the trace-free symmetric parts of the spatial variation tensor of the vorticity vanish at all epochs, it follows that this tensor must be purely antisymmetric. We now take the curl of the constraint equation (\ref{C5}), to get
\be
(\mu+p)({\rm curl}\, \om)^a=-\frac12({\rm curl}\, {\rm curl}\, q)^a\;.
\ee
 Now using the identity (\ref{id3}) and the result from Proposition 2, we get 
 \be
(\mu+p)({\rm curl}\, \om)^a=\frac12D^2q^a-\frac13\left(\mu- \frac{1}{3}\Theta^2\right)q^a\;.
\ee
\end{proof}
It is interesting to see that there exists a class of solutions with non-zero heat flux, which are the solution of the following second order differential equation
\be
\frac12D^2q^a-\frac13\left(\mu- \frac{1}{3}\Theta^2\right)q^a=0\,,
\ee
for which the vorticity can be non-zero but curl free.  
 
\section{Alternatives to gravitational waves: How silent is "silent"?}

We note that for gravitational waves to exist in a given spacetime, both the electric and magnetic parts of the Weyl tensor, $E_{ab}$ and $H_{ab}$ must be non-zero with non-zero curl. These quantities generate a tensor wave with a closed wave equation, in a similar fashion as the electric field and the magnetic field with non-zero curl generate electromagnetic vector waves in electromagnetism. Therefore in any spacetime where either $E_{ab}$ or $H_{ab}$ or their curl vanishes identically, will be devoid of any gravitational waves. Such cosmologies, that are commonly termed as {\it silent universes} (see \cite{Ellisbook} and the references therein), any information about change in local curvature of the manifold cannot causally travel via gravitational waves. 

Obviously, conformally flat spacetimes fall in the category of {\it silent universes}, as in this case the complete Weyl tensor is identically zero, and the Riemann tensor is entirely specified by the matter variables. Therefore any information about local change of curvature must causally travel via propagation of matter disturbances. The question is: {\it Can we quantify the process via which any information about local change of spacetime curvature causally travels in conformally flat models?} In this section we transparently demonstrate two such processes, a closed tensor wave equation for matter shear and a closed vector wave equation for vorticity, that can carry such information causally. 

\begin{prop}
In a conformally flat perturbation of FLRW spacetime, the shear tensor obeys a closed and transverse-traceless tensor wave equation, which is given by
\ba\label{sw}
\Box \sigma^{\la ab\ra}&\equiv &\ddot{\sigma}^{\la ab\ra}-D^2 \sigma^{\la ab\ra}=-\Theta\dot{\sigma}^{\la ab\ra}\nonumber\\
&&+\left(\frac13\Theta^2-\frac76\mu+\frac12p\right)\sigma^{\la ab\ra}\;.
\ea
\end{prop}
\begin{proof}
Since in this case we are only concerned with the tensor modes, we use the standard procedure of neglecting all first order vector perturbations, namely the gradient of background scalars together with the acceleration, heat flux and vorticity. Since the shear tensor is curl free, $({\rm curl}\, {\rm curl}\,\sigma)_{ab}=0$, and equation (\ref{id4}) becomes
\be
D^2\sigma_{ab}=\frac32D_{\la a}D^c\sigma_{b\ra c}+\left(\mu- \frac{1}{3}\Theta^2\right)\sigma_{ab}\;.
\ee
Now, the first term in the RHS is linked to vorticity, heat flux and gradient of expansion by constraint (\ref{C1}). Neglecting that term we have 
\be\label{sw1}
D^2\sigma_{ab}=\left(\mu- \frac{1}{3}\Theta^2\right)\sigma_{ab}\;.
\ee
Furthermore, taking the dot of shear evolution equation (\ref{E2}) and neglecting the acceleration term, we get
\be\label{eq2}
\ddot{\sigma}^{\la ab \ra}=\frac{1}{2}\dot{\pi}^{\la ab\ra} -\frac{2}{3}\dot{\Theta}\sigma^{ab} -\frac{2}{3}\Theta\left(\frac{1}{2}\pi^{ab} -\frac{2}{3}\Theta\sigma^{ab}\right)\;.
\ee
Using the evolution equation (\ref{E4}) and neglecting the heat flux term, we have 
\be
\dot{\pi}^{<ab>}= -(\mu +p)\sigma^{ab}- \frac{\Theta}{3}\pi^{ab}\;.
\ee
Plugging this in (\ref{eq2}), and noting that $\pi^{ab}=2\dot{\sigma}^{\la ab \ra}+\frac{4}{3}\Theta\sigma^{ab}$, we get
\be\label{sw2}
\ddot{\sigma}^{\la ab \ra}=-\frac16(\mu-3p)\sigma^{ab}-\Theta\dot{\sigma}^{\la ab \ra}\;.
\ee
Subtracting equation (\ref{sw1}) from (\ref{sw2}), we get the required tensor wave equation  (\ref{sw}).
\end{proof}
It is interesting to note that similar shear wave exists, even when the perturbations are not conformally flat, but the matter is taken to be perfect fluid, as proved in \cite{Dunsby}. What we showed here is that these waves do not go away, when we take a general form of matter perturbation and restrict the Weyl tensor to be identically zero. Also, when the expansion of the spacetime is positive, these waves get damped as they move towards the causal future. 

\begin{prop}
In a conformally flat perturbation of FLRW spacetime, if the acceleration is curl free, then the vorticity vector obeys a closed vorticity wave equation, given by
\be\label{vw}
\Box\om^{\la a \ra}\equiv\ddot{\om}^{\la a \ra}-D^2{\om}^{\la a \ra}=\left(\mu+p+\frac49\Theta^2\right)w^a\;.
\ee
\end{prop}
\begin{proof}
The proof of this proposition crucially depends on the result of Proposition 1, that is the shear tensor being curl free. In that case we can use (\ref{id2}) to write
\be
0=D^b({\rm curl}\, \sigma)_{ab}=\frac12{\rm curl}\,(D^b\sigma_{ab})\;,
\ee
which can be further simplified using the constraint (\ref{C1}), and we get
\be
\frac13 ({\rm curl}\, D_a\Theta)-\frac12({\rm curl}\, {\rm curl}\, \om)_a-\frac12({\rm curl}\, q)_a=0\;.
\ee
Now using the commutation (\ref{com12}) for the first term in the LHS, the identity (\ref{id3}) for the second term and the constraint (\ref{C5}) for the third term, we get 
\be
D^2 \om^a=\left(\frac29\Theta^2-\frac23\mu\right)\om^a\;.
\label{eqn4}
\ee
Furthermore, when the curl of the acceleration term vanishes, we have 
\be
{\dot\om}^{\la a\ra}=-\frac23\Theta\om^a\;.
\ee
Taking the dot of the above equation and using the evolution equation (\ref{E1}), we get
\be
\ddot{\om}^{\la a \ra}=\frac13\left(2\Theta^2+\mu+2p\right)\om^a\;.
\label{eqn5}
\ee
Subtracting (\ref{eqn4}) from (\ref{eqn5}), we get the required result.
\end{proof}

\section{Discussion}

Geometrical properties of general conformally flat spacetimes are still under active investigations. The aim is to understand transparently, how different geometrical and thermodynamical quantities of spacetime interact in the absence of free gravity, that is generated by the Weyl tensor. Taking the problem other way round, this understanding will definitely help us in recognising the effects of free gravity with better clarity. 
In this paper we tried to shed some light on this problem by considering a linearised but conformally flat perturbation of FLRW background, which is the well known and simplest non-trivial conformally flat solution of Einstein field equations. \\

Working in the linearised regime, we transparently demonstrated some interesting features of matter shear and vorticity and about how they are powered by different thermodynamic quantities of matter, like energy density, heat flux, isotropic pressure and anisotropic stress. These results are novel, as they clearly show the role vorticity plays in the conformally flat scenarios, as hardly any physically realistic and rotating conformally flat solutions have been found till now.  Although these results are only valid in the linearised regime, they give an indication as to how these quantities can behave in a more general setting of conformally flat spacetimes. \\

The most important point that emerged from this investigation, is that both the matter shear and the vorticity obey a transverse traceless tensor wave equation and a vector wave equation respectively. These shear and vorticity waves actually replace the gravitational waves, that these spacetimes are devoid of, in the sense that any information about local change in the curvature of the spacetime can be propagated causally via these waves. Presence of these waves makes the dynamics of relativistic and conformally flat fluid flows extremely interesting and can shed new light on the general conformally flat solutions of the Einstein field equations. 

\begin{acknowledgments}
 RM and RG are supported by National Research Foundation (NRF), South Africa. SDM 
acknowledges that this work is based on research supported by the South African Research Chair Initiative of the Department of
Science and Technology and the National Research Foundation.
\end{acknowledgments}

\end{document}